\documentclass{elsarticle}

\makeatletter
\def\ps@pprintTitle{%
 \let\@oddhead\@empty
 \let\@evenhead\@empty
 \def\@oddfoot{\centerline{\thepage}}%
 \let\@evenfoot\@oddfoot}
\makeatother

\usepackage{lmodern}
\usepackage{amsmath}
\usepackage{amsthm}
\usepackage{amsfonts}
\usepackage{amssymb}
\usepackage[dvipsnames]{xcolor} 
\usepackage[colorlinks=true,linkcolor=RoyalBlue,citecolor=PineGreen]{hyperref}
\usepackage{comment}
\usepackage{algorithm}
\usepackage{algorithmic}
\usepackage{color}

\newtheorem{Theorem}{Theorem}[section]
\newtheorem{Lemma}[Theorem]{Lemma}
\newtheorem{Proposition}[Theorem]{Proposition}

\newtheorem{Definition}[Theorem]{Definition}
\newtheorem{Remark}[Theorem]{Remark}

\theoremstyle{definition}     
\newtheorem{Example}[Theorem]{Example}

\newcommand{\N}{\mathbb{N}}
\newcommand{\K}{\mathbb{K}}
\newcommand{\I}{\mathcal{I}}
\newcommand{\J}{\mathcal{J}}
\newcommand{\V}{\mathbb{V}}

\DeclareMathOperator{\Res}{Res}
\DeclareMathOperator{\lc}{lc}

\newcommand{\ie}{{\it i.e.}}

\allowdisplaybreaks

\begin{document}

\begin{frontmatter}

\title{On Existence and Uniqueness of Formal Power Series Solutions of Algebraic Ordinary Differential Equations\tnoteref{thank}}

\author[RISC]{Sebastian Falkensteiner\corref{cor2}}
\ead{falkensteiner@risc.jku.at}
\author[XJLU]{Yi Zhang\corref{cor2}}
\cortext[cor2]{Co-first authors}
\ead{Yi.Zhang03@xjtlu.edu.cn}
\author[TDTU]{Thieu N. Vo\corref{cor1}}
\cortext[cor1]{Corresponding author}
\ead{vongocthieu@tdtu.edu.vn}

\address[RISC]{Research Institute for Symbolic Computation (RISC), Johannes Kepler University Linz, Austria}
\address[XJLU]{Department of Applied Mathematics, School of Science, Xi'an Jiaotong-Liverpool University, Suzhou, 215123, China}
\address[TDTU]{Fractional Calculus, Optimization and Algebra Research Group, Faculty of Mathematics and Statistics, Ton Duc Thang University, Ho Chi Minh City, Vietnam}

\tnotetext[thank]{This research is funded by Vietnam National Foundation for Science and Technology Development (NAFOSTED) under grant number 101.04-2019.06; supported by the Austrian Science Fund (FWF): P 31327-N32; the UTD start-up grant: P-1-03246, the Natural Science Foundation of USA grants CCF-1815108 and CCF-1708884; XJTLU Research Development Funding RDF-20-01-12.}

\begin{abstract}
Given an algebraic ordinary differential equation (AODE), we propose a computational method to determine when a truncated power series can be extended to a formal power series solution. 
If a certain regularity condition on the given AODE or on the initial values is fulfilled, we compute all of the solutions.
Moreover, when the existence is confirmed, we present the algebraic structure of the set of all formal power series solutions.
\end{abstract}


\end{frontmatter}

\section{Introduction}
The problem of finding formal power series solutions of algebraic ordinary differential equations (AODEs) has a long history and it has been extensively studied in literature. 
In~\cite{BriotBouquet}, Briot and Bouquet use the Newton polygon method to study the singularities of first-order and first degree ODEs. 
Fine gave a generalization of the method for arbitrary order AODEs in~\cite{Fine1889}. 
In~\cite{GrSi:1991}, Grigoriev and Singer proposed a parametric version of the Newton polygon method and use it to study formal power series solutions of AODEs with real exponents. 
Further study of the parametric Newton polygon method can be found in~\cite{Cano2005} and~\cite{CanoFortuny2009}. 
The Newton polygon method has its computational limits, which is illustrated in~\cite{DoraJung1997}.
For a more recent exposition about the history of the Newton polygon method, we refer to~\cite{Dragovic2019}. 
Bruno presented in~\cite{bruno2000power,bruno2004asymptotic} another method by using a polygon construction and power transformations to derive more general types of solutions such as power series with complex exponents, power-logarithmic series. 
In general, no a-priori computational bounds are presented for deciding existence and uniqueness of those solutions.
But we want to emphasize that if a certain condition is satisfied by the constructed polygon, such bounds can be given (see~\cite{Cano2005,bruno2004asymptotic}).

Van der Hoeven gave important contributions on bounding the number of initial conditions needed so that uniqueness of the formal power series solution is guaranteed (which is called ``root separation'' \cite{Hoeven2019}). 
The bounds of the root separation lead to an algorithm to decide whether a given effective formal power series is a solution of a given AODE.
However, the bounds essentially depend on the formal power series and they can be arbitrary large or even equal to infinity.
In general, one need to know a finite representation of all coefficients in order to determine such a bound.
Further study in this aspect can be found in the works of J. van der Hoeven and his collaborators, such as \cite{Aschenbrenner2017, Hoeven2006}.

In the case of AODEs of order one with constant coefficients, existence and uniqueness can be decided and all formal power series solutions can be found, see~\cite{FalkSendra}, and its generalization to formal power series solutions with fractional exponents in~\cite{CanoFalkSendra}.

This paper concerns the computation of formal power series solution of an AODE with a certain number of initial conditions.
In particular, we study the problem of deciding when a given truncated power series can be extended to a formal power series solution and compute all of them in the affirmative case. 
We follow the method inherited from the work by Denef and Lipshitz~\cite{Denef1984}.
There the authors give an expression of the derivatives of a differential polynomial with respect to the independent variable in terms of lower order differential polynomials (see~\cite[Lemma 2.2]{Denef1984} and originally \cite[page 328--329]{Hurwitz1889}). 
Our first contribution is to enlarge the class of differential equations where all formal power series solutions with a given initial value can be computed algorithmically. 
This class is given by a sufficient condition on the given differential equation and initial value which is described by the local vanishing order. 
Moreover, we give a necessary and sufficient condition on the given differential equation such that for every initial value all formal power series solutions can be computed in this way. 
For differential equations satisfying this condition, we give an algorithm to compute all formal power series solutions up to an arbitrary order and illustrate it by some examples.

We note that Denef and Lipshitz provided an algorithmic method to decide the existence of an AODE with initial conditions in \cite{Denef1984}.
However, the computation of all such solutions was not the concern of their work.
Moreover, the solutions were assumed to be non-singular, and the singular solutions were avoided.
In our method, both singular and non-singular solutions are computed.
We also note that some of the intermediate results and reasonings in our work could be derived by using the Newton polygon method or other polygon constructions, but we follow a completely different strategy and give independent proofs.

The rest of the paper is organized as follows. 
Section~\ref{sec-IFT} is devoted to present a well-known formula of Ritt (Lemma~\ref{LEM:Ritt}) which can be used for finding formal power series solutions by using coefficient comparison, see Proposition~\ref{PROP:ImplicitFunctionTheorem}.
Since not all formal power series solutions can be found in this way (For instance, see Example~\ref{EX:mnonvanishing}), one may use a refinement of Ritt's formula presented in~\cite{Hurwitz1889, limonov2015generalized}. 
We summarize it by Theorem~\ref{THM:generalizedSeparantFormula} in Section~\ref{sec-generalized-separants}. 
In order to simplify some of the subsequent reasonings, we also use a slightly different notation and define \textit{separant matrices}. 
Moreover, we give some sufficient conditions on the given differential equation, which is called the \textit{vanishing order}, such that the refined formula can be used in an algorithmic way for computing all formal power series solutions and present new results in this direction, see Theorems~\ref{THM:local}, \ref{THM:global} and Algorithm~\ref{ALGO:localnonvanishing}. 
In Section~\ref{subsec-locally} we focus on solutions with given initial values and study the vanishing order locally, whereas in Section~\ref{subsec-globally} we generalize the vanishing order to arbitrary initial values. 
For the global situation, Proposition~\ref{PROP:specialnonvanishing} and~\ref{PROP:charVanishing} show that a large class of AODEs indeed satisfy our sufficient conditions.

\section{Implicit Function Theorem for AODEs} \label{sec-IFT}
Let $\K$ be an algebraically closed field of characteristic zero.
For the computational purpose, we can choose $\K$ to be the field $\overline{\mathbb{Q}}$ of algebraic numbers or the field $\mathbb{C}$ of complex numbers. 
We consider an algebraic ordinary differential equation (AODE) of the form 
\begin{equation}\label{eq:AODE}
F(x,y,y',\ldots,y^{(n)})=0,
\end{equation}
where $n \in \N$ and $F$ is a differential polynomial in $\K[x]\{y\}$ of order $n$. 
For simplicity we may also write~\eqref{eq:AODE} as $F(y) = 0$ and call $n$ the \textit{order} of~\eqref{eq:AODE}.

Let $\K[[x]]$ be the ring of formal power series with respect to $x$ around the origin. 
For each formal power series $f \in \K[[x]]$ and $k \in \N$, we use the notation $[x^k] f$ to refer to the coefficient of $x^k$ in $f$. 
The coefficient of $x^k$ in a formal power series can be expressed by means of the constant coefficient of its $k$-th formal derivative, as stated in the following lemma (see~\cite[Theorem 2.3, page 20]{KauersPaule2010}).  
\begin{Lemma} \label{LEM:kthcoeff}
Let $f \in \K[[x]]$ and $k \in \N$. 
Then $[x^k] f = [x^0] \left(\frac{1}{k!}\,f^{(k)} \right)$.
\end{Lemma}

Let $n \in \N$. 
We define the projection map $\pi_n$ as
\[
\begin{array}{cccc}
\pi_{n}: & \K^{\N} & \longrightarrow & \K^{n+1} \\
& (c_0, c_1, \ldots) & \longmapsto & (c_0, \ldots, c_n).
\end{array}
\]

Assume that $y(x)=\sum_{i \geq 0} \frac{c_i}{i!}\,x^i \in \K[[x]]$ is a formal power series solution of $F(y) = 0$ around the origin, where $c_i \in \K$ are unknowns and $F$ is of order $n$. 
Set $\textbf{c}=(c_0,c_1,\ldots) \in \K^\N$. 
By Lemma~\ref{LEM:kthcoeff}, we know that $F(y(x))=0$ if and only if $$[x^0](F^{(k)}(y(x))) = F^{(k)}(0,\pi_{n+k}(\textbf{c}))=0 \text{ for each } k \geq 0.$$
The above fact motivates the following definition.

\begin{Definition} \label{DEF:jet}
Let $F \in \K[x]\{y\}$ be a differential polynomial of order $n \in \N$. 
\begin{enumerate}
	\item Assume that $\textbf{c}=(c_0,c_1,\ldots) \in \K^\N$ is an infinite tuple of indeterminates and $m \in \N$. 
	We call the ideal 
	$$\mathcal{J}_m(F)=\langle F(0,\pi_{n}(\textbf{c})),\ldots,F^{(m)}(0,\pi_{n+m}(\textbf{c})) \rangle \subseteq \K[c_0,\dots,c_{n+m}]$$
	 the $m$-th \textit{jet ideal} of $F$. 
	Moreover, we denote the zero set of $\mathcal{J}_m(F)$ by $\mathbb{V}(\mathcal{J}_m(F)) \subseteq \K^{n+m+1}$.
	\item Let $k \in \N$. 
	Assume that $\tilde{\textbf{c}}=(c_0,c_1,\ldots,c_{k})\in \K^{k+1}$ and $\tilde{y}(x)= \sum_{i=0}^k \frac{c_{i}}{i!}\,x^{i}$.
	We say that $\tilde{\textbf{c}}$, or $\tilde{y}(x)$, can be extended to a formal power series solutions of $F(y)=0$ if there exists $y(x) \in \K[[x]]$ such that $F(y(x)) = 0$ and $$y(x) \equiv \tilde{y}(x) \mod x^{k+1}.$$
\end{enumerate}
\end{Definition}

By the above definition, we know that $y(x) \in \K[[x]]$ is a solution of $F(y)=0$ if and only if $\pi_{n+k}(\textbf{c}) \in \V(\J_k(F))$ for each $k \geq 0$, where $\textbf{c} \in \K^{\N}$ denotes the infinite tuple of coefficients of $y(x)$. 
Since it is impossible to check the latter relation for each $k \in \N$, we need to find an upper bound $k_0$ such that $\pi_{n+k}(\textbf{c}) \in \V(\J_k(F))$ for each $k \ge k_0$ or $\pi_{n+k_0}(\textbf{c}) \not\in \V(\J_{k_0}(F))$.
The following example shows that in general this bound, if it exists, can be arbitrarily large.

\begin{Example} \label{EX:mnonvanishing}
For each positive integer $m$, consider the AODE $$F=x\,y'-m\,y+x^{m} = 0$$ with the initial tuple $(c_0,c_1)=(0,0) \in \V(\J_0(F))$. 
For every $0 \leq k<m$, we have $$F^{(k)}=x\,y^{(k+1)}+(k-m)\,y^{(k)}+m(m-1) \cdots (m-k+1)\,x^{m-k}.$$ 
Therefore, we have that $(c_0,\ldots,c_{k+1}) \in \V(\J_{k})$ for all $0 \leq k<m$ if and only if $c_0=\cdots=c_{k}=0$. 
However, $$F^{(m)}(0,\ldots,0,c_{m},c_{m+1})=m! \neq 0$$ and $(c_0,c_1)=(0,0)$ cannot be extended to a formal power series solution of $F(y)=0$.
\end{Example}

Note that Example~\ref{EX:mnonvanishing} does not exclude the existence of such an upper bound in terms of the coefficients and exponents of $F$. 
Let us first recall a lemma which shows that for $k \geq 1$ the highest occurring derivative appears linearly in the $k$-th derivative of $F$ with respect to $x$ (see~\cite[page 30]{ritt1950differential}).

\begin{Lemma} \label{LEM:Ritt}
Let $F \in \K[x]\{y\}$ be a differential polynomial of order $n \geq 0$. 
Then for each $k \geq 1$, there exists a differential polynomial $R_k \in \K[x]\{y\}$ of order at most $n+k-1$ such that
\begin{equation} \label{EQ:separant}
F^{(k)}=S_F \cdot y^{(n+k)} + R_{k},
\end{equation}
where $S_F=\frac{\partial F}{\partial y^{(n)}}$ is the separant of $F$. 
\end{Lemma}

Based on Lemma~\ref{LEM:Ritt} and the reasonings in the beginning of the section, we have the following proposition, which is sometimes called Implicit Function Theorem for AODEs as a folklore.

\begin{Proposition} \label{PROP:ImplicitFunctionTheorem} 
Let $F \in \K[x]\{y\}$ be a differential polynomial of order $n$. 
Assume that $\tilde{\textbf{c}}=(c_0,c_1,\ldots,c_n) \in \V(\J_0(F))$ and $S_F(\tilde{\textbf{c}}) \neq 0$. 
For $k>0$, set $$c_{n+k} = -\frac{R_k(0,c_0,\ldots,c_{n+k-1})}{S_F(\tilde{\textbf{c}})},$$
where $R_k$ is specified in Lemma~\ref{LEM:Ritt}.
Then $y(x)=\sum_{i \geq 0} \frac{c_i}{i!}\,x^i$ is a solution of $F(y)=0$.
\end{Proposition}

In the above proposition, if the initial value $\tilde{\textbf{c}}$ vanishes at the separant of $F$, we may expand $R_k$ in Lemma~\ref{LEM:Ritt} further in order to find formal power series solutions. 
In the forthcoming sections, we will develop this idea in a systematical way.

\section{Generalized separants} \label{sec-generalized-separants}
In~\cite{limonov2015generalized} the author presents an expansion formula for derivatives of $F$ with respect to $x$ showing that not only the highest occurring derivative appears linearly, also the second-highest one, third-highest one, and so on. 
This generalizes Lemma~\ref{LEM:Ritt} and refines~\cite{Denef1984}[Lemma 2.2].


\begin{Definition}[{See \cite[Definition~8]{limonov2015generalized}}] \label{DEF:generalizedSeparant}
For a differential polynomial $F \in \K[x]\{y\}$ of order $n \in \N$ and $k,i \in \N$, we define $$f_{i} = \begin{cases} \frac{\partial \,F}{\partial y^{(i)}}, &i = 0,\ldots,n;\\ 0, &\text{otherwise}; \end{cases}$$ and $$S_{F,k,i}=\sum_{j=0}^{i} \binom{k}{j}f_{n-i+j}^{(j)}.$$
We call $S_{F,k,i}$ the \textit{generalized separants} of $F$.
\end{Definition}

Note that $S_{F,k,0}$ coincides with the usual separant $\frac{\partial F}{\partial y^{(n)}}$ of $F$ and the order of $S_{F,k,i}$ is less than or equal to $n+i$.

\begin{Lemma} \label{THM:generalizedSeparantFormula}
Let $F \in \K[x]\{y\}$ be a differential polynomial of order $n \in \N$. 
Then for each $m \in \N$ and $k>2m$ there exists a differential polynomial $r_{n+k-m-1}$ with order less than or equal to $n+k-m-1$ such that 
\begin{equation} \label{EQ:SumForm}
	F^{(k)}=\sum_{i=0}^{m} S_{F,k,i}\,y^{(n+k-i)} + r_{n+k-m-1}.
\end{equation}
\end{Lemma}
\begin{proof}
See \cite{limonov2015generalized}[Corollary 1].
\end{proof}

Let $F \in \K[x]\{y\}$ of order $n$ and $k > 2m$. We define

\[
\mathcal{B}_{m}(k)=
\begin{bmatrix}
\binom{k}{0} & \binom{k}{1} & \ldots & \binom{k}{m}
\end{bmatrix},
\]
and 
\[
\mathcal{S}_{F,m}=
\begin{bmatrix}
f_n & f_{n-1} & f_{n-2} & \cdots & f_{n-m}\\
0 & f_n^{(1)} & f_{n-1}^{(1)} & \cdots & f_{n-m+1}^{(1)}\\
0 & 0 & f_n^{(2)} & \cdots & f_{n-m+2}^{(2)}\\
\vdots & \vdots & \vdots & \vdots & \vdots\\
0 & 0 & 0 & \cdots & f_n^{(m)}
\end{bmatrix},
\]
and
$$Y_m=
\begin{bmatrix}
y^{(m)}\\y^{(m-1)}\\\vdots\\y
\end{bmatrix}.
$$

We call $\mathcal{S}_{F,m}$ the \textit{$m$-th separant matrix of $F$}.

Then we can also represent formula~\eqref{EQ:SumForm} of Lemma~\ref{THM:generalizedSeparantFormula} as
\begin{align} \label{EQ:MatrixForm}
F^{(k)}=\mathcal{B}_{m}(k) \cdot \mathcal{S}_{F,m} \cdot Y_m^{(n+k-m)} + r_{n+k-m-1}.
\end{align}

It is straightforward to see that for $m \in \N$ and $\textbf{c} \in \K^{\N}$ the separant matrix $\mathcal{S}_{F,m}(\textbf{c})=0$ if and only if $S_{F,k,i}(\textbf{c})=0$ holds for all $0 \leq i \leq m$ and $k \in \N$.
In the following we will look for non-zero entries of $\mathcal{S}_{F,m}(\textbf{c})$ or equivalently, of $S_{F,k,i}(\textbf{c})$.

\begin{Remark} \label{rem-localnonvanishing}
Recall that a solution of $F(y)=0$ is called \textit{non-singular} if it does not vanish at the separant $S_F=\frac{\partial F}{\partial y^{(n)}}$.
If a formal power series $y(x)=\sum_{i \geq 0} \frac{c_i}{i!}\,x^i \in \K[[x]]$ is a non-singular solution of $F(y)=0$, then there exists $m \in \N$ such that $S_F^{(m)}(0,c_0,\ldots,c_{n+m}) \neq 0$. 
Hence, $\mathcal{S}_{F,m}(0,c_0,\ldots,c_{n+m}) \neq 0$.
\end{Remark}

The Remark~\ref{rem-localnonvanishing} is one of the main points in the proofs of the statements in~\cite{Denef1984}. 
Let us briefly describe the most important results related to our work:
\begin{itemize}
	\item In Theorem~2.7 in \cite{Denef1984}, the Strong Approximation Theorem, it is shown that for a sufficiently large number of given initial values the existence of a solution immediately follows.
	\item Theorem~3.1 in \cite{Denef1984} shows that the existence of a formal power series solution of (systems of) equation can be decided algorithmically. 
\end{itemize}
Both results are focusing on the existence of a formal power series solution. 
In fact, in the proof of Theorem 3.1 they use an inequality which essentially states that, for a given initial tuple $\textbf{c} \in \V(\J_{2m})$, there is a prolongation such that $S_{F,m}(\textbf{c}) \ne 0$. 
Hence, uniqueness of the solution is not given in the arising procedure. 

Another frequently used assumption in Differential Algebra is that formal power series solutions are assumed to be non-singular (see for example in the proof of~\cite{Denef1984}[Theorem 2.7]). 
Let us first recall that this assumption is not restrictive from a theoretical point of view.
\begin{Remark} \label{rem-regular}
Let $y(x)$ be a (non-constant) singular solution of $F=0$ where $F \in \K[x,y,\ldots,y^{(n)}]$. 
Then $y(x)$ is also a zero of the resultant $$F_1=\Res_{y^{(n)}}(F,S_F) \in \K[x,y,\ldots,y^{(n-1)}].$$
Now $y(x)$ might be again a singular solution of an irreducible component of $F_1$, say $\tilde{F}_1$ and hence, it is a zero of $F_2=\Res_{y^{(n-1)}}(\tilde{F}_1,S_{\tilde{F}_1})$. 
Continuing this procedure, we obtain $n+1$ AODEs, say $F_1,F_2,\ldots,F_n,$ defined by
\begin{align*}
F_i=\Res_{y^{(n-i+1)}}(\tilde{F}_{i-1},S_{\tilde{F}_{i-1}}), \quad i=1,\ldots,n.
\end{align*}
These AODEs are of strictly decreasing orders and since $F_{n}$ is a polynomial in $x$ and $y$, the irreducible components do not have (non-trivial) singular solutions.
Hence, $y(x)$ is a non-singular solution of some $\tilde{F}_{i}$ with $1 \le i<n$.
\end{Remark}
Algorithmically the procedure from Remark~\ref{rem-regular} has a problem: In general it cannot be decided whether a solution is singular or non-singular. 
Denef and Lipshitz avoid this issue by additionally proposing the inequality $S_{F,m}(\textbf{c}) \ne 0$ in the given differential problem.

By Remark~\ref{rem-localnonvanishing} and Remark~\ref{rem-regular}, for a given differential equation $F=0$, for every formal power series solution $y(x)=\sum_{i \geq 0} \frac{c_i}{i!}\,x^i$ there is $m \in \N$ and some differential polynomial of $F=\tilde{F}_0,\tilde{F}_1,\ldots,\tilde{F}_{n}$ obtained from $F$ as above, such that $$S_{\tilde{F}_i,m}(0,c_0,\ldots,c_{n+m}) \ne 0.$$
Since this number $m$ is in principle unbounded, we focus in this work on the given equation $F=0$ and look for non-singular and singular solutions at the same time. 
This is an essential difference to the approach in~\cite{Denef1984}.

\section{Local vanishing order} \label{subsec-locally}
In this section, we consider the problem of deciding when a solution modulo a certain power of $x$ of a given AODE can be extended to a full formal power series solution. 
By using Lemma~\ref{THM:generalizedSeparantFormula}, we present a partial answer for this problem. 
In particular, given a certain number of coefficients satisfying some additional assumptions, we propose an algorithm to check whether there is a formal power series solution whose first coefficients are the given ones, and in the affirmative case, compute all of them (see Theorem~\ref{THM:local} and Algorithm~\ref{ALGO:localnonvanishing}).

Let us start with a technical lemma which we will frequently use later.

\begin{Lemma} \label{LEM:2mvanish}
Let $m,n \in \N$ and $F \in \K[x]\{y\}$ be a differential polynomial of order $n$ and $\textbf{c} = (c_0,c_1,\ldots)$ be a infinite tuple of indeterminates. 
Assume that the generalized separants $S_{F,k,i}(\textbf{c})=0$ for all $0 \leq i<m$.
Then the differential polynomial $F^{(k)}(\textbf{c})$ involves only
\begin{enumerate}
	\item $c_0,\ldots,c_{n+\lfloor k/2 \rfloor}$ for $0 \leq k \leq 2m$;
	\item $c_0,\ldots,c_{n+k-m}$ for $k>2m$.
\end{enumerate}
\end{Lemma}
\begin{proof}
Assume that $0 \leq k \leq 2m$.  Set $\tilde{m}= k-\lfloor k/2 \rfloor$. 
Then $k>2 \tilde{m}-1$. By assumption and Theorem~\ref{THM:generalizedSeparantFormula}, we have 
$$F^{(k)}(\textbf{c})= r_{n+k-\tilde{m}}(\textbf{c}),$$ where $r_{n+k-\tilde{m}}$ is a differential polynomial of order at most $n+k-\tilde{m}=n+\lfloor k/2 \rfloor$.
Thus, item 1 follows.

Let $k>2m$. 
By~\eqref{EQ:MatrixForm} and the assumption, we get 
$$F^{(k)}(\textbf{c})= S_{F,k,m}(\textbf{c})\,c_{n+k-m} + r_{n+k-m-1}(\textbf{c}),$$ and thus item 2 holds.
\end{proof}

\begin{Definition} \label{DEF:LocalVanishing}
Let $m,n \in \N$ and $F \in \K[x]\{y\}$ be a differential polynomial of order $n$. 
Let $\tilde{\textbf{c}}=(c_0,\dots,c_{n+m}) \in \K^{n+m+1}$. 
We say that $F$ has \textit{vanishing order $m$ at $\tilde{\textbf{c}}$} if the following conditions hold:
\begin{enumerate}
	\item \label{DEF:LocalVanishingItem1} $\mathcal{S}_{F,i}(\tilde{\textbf{c}}) = 0$ for all $0 \leq i<m$, and $\mathcal{S}_{F,m}(\tilde{\textbf{c}}) \neq 0$;
	\item \label{DEF:LocalVanishingItem2} $\tilde{\textbf{c}} \in \V(\J_{2m}(F))$.
\end{enumerate}
\end{Definition}

As a consequence of Lemma~\ref{LEM:2mvanish} and item \ref{DEF:LocalVanishingItem1} of the above definition, $\V(\J_{2m}(F))$ can be seen as a subset of $\K^{n+m+1}$ and therefore item \ref{DEF:LocalVanishingItem2} is well-defined.

Let $m,n \in \N$ and $F \in \K[x]\{y\}$ be a differential polynomial of order $n$. 
Assume that $\textbf{c}=(c_0,\ldots,c_{n+m}) \in \K^{n+m+1}$ and $F$ has vanishing order $m$ at $\textbf{c}$. 
We regard $S_{F,t,m}(\textbf{c})$ as a polynomial in $t$ and denote
\begin{align*}
\textbf{r}_{F,\textbf{c}}&= \text{the number of integer roots of } S_{F,t,m}(\textbf{c}) \text{ which are greater than } 2m,\\
\textbf{q}_{F,\textbf{c}}&= 
\begin{cases}
\text{the largest integer root of } S_{F,t,m}(\textbf{c}), & \text{ if } \textbf{r}_{F,\textbf{c}} \geq 1,\\
2m, & \text{ otherwise.}
\end{cases}
\end{align*}

\begin{Theorem} \label{THM:local}
Let $F \in \K[x]\{y\}$ be of vanishing order $m \in \N$ at $\textbf{c} \in \K^{n+m+1}$.
\begin{enumerate}
	\item Then $\textbf{c}$ can be extended to a formal power series solution of $F(y)=0$ if and only if it can be extended to a zero of $\J_{\textbf{q}_{F,\textbf{c}}}(F)$.
	\item Let \begin{align*}
	\mathcal{V}_{\textbf{c}}(F)= \pi_{n+\textbf{q}_{F,\textbf{c}}-m}(\{\tilde{\textbf{c}} \in \mathbb{V}(\mathcal{J}_{\textbf{q}_{F,\textbf{c}}}(F)) ~|~ \pi_{n+m}(\tilde{\textbf{c}}) = \textbf{c} \}).
	\end{align*}
	Then $\mathcal{V}_{\textbf{c}}(F)$ is an affine variety of dimension at most $\textbf{r}_{F,c}$. 
	Moreover, each point of it can be uniquely extended to a formal power series solution of $F(y)=0$.
\end{enumerate}
\end{Theorem}

\begin{proof}
1. Let $\bar{\textbf{c}}=(c_0,c_1,\ldots) \in \K^{\N}$ be such that $\pi_{n+m}(\bar{\textbf{c}})=\textbf{c}$, where $c_k$ is to be determined for $k>n+m$.
Recall that $\sum_{i \geq 0} \frac{c_i}{i!} \, x^i$ is a solution of $F(y)=0$ if and only if $F^{(k)}(\bar{\textbf{c}}) = 0$ for each $k>2m$. 
Since $F$ has a vanishing order $m$, and by Theorem~\ref{THM:generalizedSeparantFormula}, there is a differential polynomial $r_{n+k-m-1}$ of order at most $n+k-m-1$ such that
\begin{align} \label{EQ:equivalent2}
F^{(k)}(\bar{\textbf{c}}) = S_{F,k,m}(\textbf{c})\,c_{n+k-m}+r_{n+k-m-1}(\bar{\textbf{c}})=0.
\end{align}

If $\textbf{c}$ can be extended to a solution of $F(y) = 0$, then it follows from Definition~\ref{DEF:jet} that $\textbf{c}$ can be extended to a zero of $\mathcal{J}_{\textbf{q}_{F,\textbf{c}}}(F)$. 
Vice versa, if $\textbf{c}$ can be extended to a zero of $\mathcal{J}_{\textbf{q}_{F,\textbf{c}}}(F)$, then there exist $c_{n+m+1},\ldots,c_{n+\textbf{q}_{F,\textbf{c}}-m} \in \K$ such that $F^{(k)}(\bar{\textbf{c}}) =0$ for $k=2m,\ldots,\textbf{q}_{F,\textbf{c}}$. 
For $k > \textbf{q}_{F,\textbf{c}}$, we set 
\begin{equation} \label{EQ:equivalent3}
c_{n+k-m}=-\frac{r_{n+k-m-1}(0,c_0,\ldots,c_{n+k-m-1})}{S_{F,k,m}(\textbf{c})}
\end{equation}
and thus $y(x)=\sum_{i \geq 0} \frac{c_i}{i!} \, x^i$ is a solution of $F(y) = 0$.

2. By item 1, $\mathcal{V}_{\textbf{c}}(F)$ is the set of points satisfying
\begin{itemize}
	\item[(i)] $\tilde{c}_k = c_k  \ \ \text{ for } 0 \leq k \leq n+m;$
	\item[(ii)] $S_{F,k,m}(\textbf{c})\,\tilde{c}_{n+k-m}+r_{n+k-m-1}(\tilde{c}_0,\ldots,\tilde{c}_{n+k-m-1})=0 \ \ \text{ for } 2m <k \leq \textbf{q}_{F,\textbf{c}},$
\end{itemize}
and therefore it is an affine variety.

If $\textbf{r}_{F,\textbf{c}}=0$, then $\textbf{q}_{F, \textbf{c}}=2m$ and thus $\mathcal{V}_{\textbf{c}}(F) = \{\textbf{c}\}$ contains one point exactly.

Assume that $\textbf{r}_{F, \textbf{c}} \geq 1$. 
Let $k_1 < \cdots < k_{\textbf{r}_{F,\textbf{c}}} = \textbf{q}_{F, \textbf{c}}$ be integer roots of $S_{F,t,m}(\textbf{c})$ which are greater than $2m$.
If $k \notin \{k_1,\ldots,k_{\textbf{r}_{F,\textbf{c}}}\}$, then it follows from \eqref{EQ:equivalent2} that $\tilde{c}_{n+k-m}$ is uniquely determined from the previous coefficients and
\[
\begin{array}{cccc}
\phi: & \mathcal{V}_{\textbf{c}}(F)  & \longrightarrow & \K^{\textbf{r}_{F, \textbf{c}}} \\
&  \tilde{\textbf{c}} & \longmapsto & (\tilde{c}_{n+k_1-m},\ldots,\tilde{c}_{n+k_{\textbf{r}_{F,c}}-m})
\end{array}
\]
defines an injective map. 
Therefore, we conclude that $\mathcal{V}_{\textbf{c}}(F)$ is of dimension at most~$\textbf{r}_{F, \textbf{c}}$. 
Moreover, it follows from item 1 that each point of $\mathcal{V}_{\textbf{c}}(F)$ can be uniquely extended to a formal power series solution of $F(y) = 0$. 
\end{proof}

The proof of the above theorem is constructive. 
More precisely, if a tuple $\textbf{c} \in \K^{n+m+1}$ satisfies the condition that $F$ has vanishing order $m$ at $\textbf{c}$, then the proof gives an algorithm to decide whether $\textbf{c}$ can be extended to a formal power series solution of $F(y)=0$ or not, and in the affirmative case determine all of them. 
We summarize them as in Algorithm~\ref{ALGO:localnonvanishing}. 

\begin{algorithm}
	\caption{DirectMethodLocal}
	\label{ALGO:localnonvanishing}
	\begin{algorithmic}[1]
		\REQUIRE $\ell \in \N$, $\textbf{c}=(c_0,\ldots,c_{n+m}) \in \K^{n+m+1}$, and a differential polynomial $F$ of order $n$ which has vanishing order $m$ at $\textbf{c}$.
		\ENSURE All formal power series solutions of $F(y) = 0$, with $\textbf{c}$ as initial tuple, described as the truncation of the series up to order $\ell$ including a finite number of indeterminates and the algebraic conditions on these indeterminates. 
		The truncations are in one-to-one correspondence to the formal power series solutions.
		\STATE Compute $S_{F,k,m}(\textbf{c}), \textbf{r}_{F,\textbf{c}}, \textbf{q}_{F,\textbf{c}}$ and the defining equations of $\mathcal{V}_{\textbf{c}}(F)$:
		\begin{itemize}
			\item[(i)] $\tilde{c}_k = c_k  \ \ \text{ for } 0 \leq k \leq n+m;$
			\item[(ii)] $S_{F,k,m}(\textbf{c})\,\tilde{c}_{n+k-m}+r_{n+k-m-1}(\tilde{c}_0,\ldots,\tilde{c}_{n+k-m-1})=0 \ \ \text{ for } 2m <k \leq \textbf{q}_{F,\textbf{c}},$
		\end{itemize}
		where $\tilde{c}_k$'s are indeterminates.
		\STATE Check whether $\mathcal{V}_{\textbf{c}}(F)$ is empty or not by using Gr\"{o}bner bases. 
		\IF{$\mathcal{V}_{\textbf{c}}(F)=\emptyset$}
		\STATE Output the string ``$\textbf{c}$ can not be extended to a formal power series solution of $F(y) = 0$''. 
		\ELSE \STATE Compute $\tilde{c}_{n+\textbf{q}_{F,\textbf{c}}-m+1},\ldots,\tilde{c}_{\ell}$ by using \eqref{EQ:equivalent3}.
		\RETURN $\sum_{i=0}^\ell \frac{\tilde{c}_i}{i!}\,x^i$ and $\mathcal{V}_{\textbf{c}}(F)$.
		\ENDIF
	\end{algorithmic}
\end{algorithm}

The termination of the above algorithm is evident. 
The correctness follows from Theorem~\ref{THM:local}.

\begin{Example} \label{EX:localfps1}
Consider the following AODE of order two: $$F = x\,y'' - 3y' + x^2y^2 = 0.$$
Let $\textbf{c} = (c_0,0,0,2c_0^2) \in \pi_{3}(\mathbb{V}(\mathcal{J}_2(F)))$, where $c_0$ is an arbitrary constant in $\K$. 
One can verify that each point of $\pi_{3}(\mathbb{V}(\mathcal{J}_2(F)))$ is of the form $\textbf{c}$. 
A direct calculation shows that $F$ has vanishing order $1$ at $\textbf{c}$. 
Moreover, we have that $S_{F,t,1}(\textbf{c}) = t-3$ and $\textbf{q}_{F,\textbf{c}}= 3$.
Thus, it follows that $$\mathcal{V}_{\textbf{c}}(F) = \{ \tilde{\textbf{c}} = (c_0,0,0,2c_0^2,c_4) \in \K^5 \mid c_4 \in \K \} .$$ 
So, the dimension of $\mathcal{V}_{\textbf{c}}(F)$ is equal to one and the corresponding formal power series solutions are 
\begin{equation}\label{eq:solution}
y(x) \equiv c_0 + \frac{c_0^2}{3}\,x^3 + \frac{c_4}{24}\,x^4 - \frac{c_0^3}{18}\,x^6 - \frac{c_0\,c_4}{252}\,x^7 - \frac{c_0^2\,c_4}{3024}\,x^{10} \mod x^{11}.
\end{equation}
Above all, the set of formal power series solutions of $F(y) = 0$ at the origin is in bijection with $\K^2$  and can be represented as in~\eqref{eq:solution}.
\end{Example}

\section{Global vanishing order} \label{subsec-globally}
The input specification of Algorithm~\ref{ALGO:localnonvanishing} is that the given initial tuple $\textbf{c}$ is of length $n+m+1$ and that the differential polynomial $F$ has vanishing order $m$ at $\textbf{c}$. 
In general, the natural number $m$ can be arbitrarily large.
In this section, we give a necessary and sufficient condition for differential polynomials for which the existence for an upper bound of $m$ is guaranteed. 
If this condition is satisfied, Algorithm~\ref{ALGO:localnonvanishing} can be applied to every initial tuple of appropriate length.

\begin{Definition} \label{DEF:globalvanishing}
Let $F \in \K[x]\{y\}$ be a differential polynomial of order $n \in \N$. 
Assume that $\textbf{c}=(c_0,c_1,\ldots)$ is an infinite tuple of indeterminates and $m \in \N$.
\begin{enumerate}
	\item We define $\mathcal{I}_m(F) \subseteq \K[c_0,\ldots,c_{n+m}]$ to be the ideal generated by the entries of the separant matrix $\mathcal{S}_{F,m}(\textbf{c})$.
	\item We say that $F$ has \textit{vanishing order $m$} if $m$ is the smallest natural number such that $1 \in \I_m(F)+\J_{2m}(F) \subseteq \K[c_0,\ldots,c_{n+2m}]$, where $\J_{2m}(F)$ is the $2m$-th jet ideal of $F$. 
	If there does not exist such $m \in \N$, then we define the vanishing order of $F$ to be $\infty$.
\end{enumerate}
\end{Definition}

\begin{Proposition} \label{PROP:specialnonvanishing}
Every differential polynomial of the form $$A(x)\,y^{(m)}+B(x,y,\ldots,y^{(m-1)},y^{(m+1)},\ldots,y^{(n)})$$ has vanishing order of at most $\deg_x(A)+n-m$.
\end{Proposition}
\begin{proof}
The separant matrix $\mathcal{S}_{F,\deg_x(A)+n-m}$ has the non-zero entry 
$$f_{m}^{(\deg_x(A))}= \deg_x(A)!\,\lc(A).$$
Hence, $1 \in \I_{\deg_x(A)+n-m} \subseteq \I_{\deg_x(A)+n-m}+\J_{2\deg_x(A)+n-m}.$
\end{proof}

Proposition~\ref{PROP:specialnonvanishing} is a generalization of~\cite{limonov2015generalized}[Corollary 2] where only the case $A(x) \in \K$ is treated. 
The following proposition gives a characterization of differential polynomials with finite vanishing order.

\begin{Lemma}\label{lemma:strong_approximation}
Let $F_1,\ldots,F_k$ be differential polynomials in $\mathbb{K}(x)\{y\}$.
Suppose that for every integer $n \geq 1$,  there exists $z(x) \in \mathbb{K}[[x]]$ such that 
$$F_1(z(x))=\ldots=F_k(z(x))=0 \mod x^n,$$
then there exists $\bar{z}(x) \in \mathbb{K}[[x]]$ such that 
$$F_1(\bar{z}(x))=\ldots=F_k(\bar{z}(x))=0.$$
\end{Lemma}
\begin{proof}
This lemma is a special case of the Strong Approximation Theorem~\cite[Theorem~2.10]{Denef1984}.
\end{proof}

\begin{Proposition} \label{PROP:charVanishing}
Let $F \in \K[x]\{y\}$ be a differential polynomial of order $n$. 
Then $F$ has finite vanishing order if and only if the differential system
\begin{equation}\label{sys:FdF}
F = \frac{\partial F}{\partial y} = \cdots =\frac{\partial F}{\partial y^{(n)}} =0
\end{equation}
has no solution in $\K[[x]]$.
In particular, if the differential ideal $\left[ F, \frac{\partial F}{\partial y}, \ldots, \frac{\partial F}{\partial y^{(n)}} \right]$ in $\K(x)\{y\}$ contains $1$, then $F$ has finite vanishing order.
\end{Proposition}
\begin{proof}
Assume that the system~\eqref{sys:FdF} has a solution $y(x)=\sum_{i \geq 0} \frac{c_i}{i!} \, x^i \in \K[[x]]$. 
Then for every $m \in \N$ we have $(c_0,\ldots,c_{n+m}) \in \V(\mathcal{I}_m(F)+\mathcal{J}_{2m}(F))$.
Therefore, it follows from Hilbert's weak Nullstellensatz that $1 \notin \mathcal{I}_m(F)+\mathcal{J}_{2m}(F)$ and thus, $F$ does not have finite vanishing order.

Conversely, assume that $F$ does not have finite vanishing order,~\ie, for each $k \in \N$ we have $1 \notin \I_{k+n}(F)+\J_{2(k+n)}(F)$ and the ideals have a common root $\textbf{c}=(c_0,\ldots,c_{3n+2k}) \in \K^{3n+2k+1}$. 
In particular, we have for all $i=0,\ldots,k$
\begin{equation}\label{eq:Fc}
F^{(i)}(\textbf{c}) = \left(\frac{\partial F}{\partial y}\right)^{(i)}(\textbf{c}) = \ldots = \left(\frac{\partial F}{\partial y^{(n)}}\right)^{(i)}(\textbf{c}) = 0,
\end{equation}
and therefore, by Lemma~\ref{LEM:kthcoeff}, for $\tilde{y}(x)=\sum_{i=0}^{3n+2k} \frac{c_i}{i!} \, x^i$ and every $i=0,\ldots,k$ also
\begin{equation*}
[x^i]F(\tilde{y}(x)) = [x^i] \frac{\partial F}{\partial y} (\tilde{y}(x)) = \cdots = [x^i] \frac{\partial F}{\partial y^{(n)}} (\tilde{y}(x)) = 0.
\end{equation*}
Thus, $\tilde{y}(x)$ is a solution of $F=\frac{\partial F}{\partial y}=\cdots=\frac{\partial F}{\partial y^{(n)}}=0$ modulo $x^{k}$. 
Due to Lemma~\ref{lemma:strong_approximation}, we conclude that the system~\eqref{sys:FdF} admits a solution in $\K[[x]]$ and the equivalence is proven.

Assume that $1 \in \left[ F, \frac{\partial F}{\partial y}, \ldots, \frac{\partial F}{\partial y^{(n)}} \right]$. 
Then system~\eqref{sys:FdF} does not have a solution in $\K[[x]]$ and thus $F$ has finite vanishing order.
\end{proof}

\begin{Remark}
To test whether $1 \in \left[ F, \frac{\partial F}{\partial y}, \ldots, \frac{\partial F}{\partial y^{(n)}} \right]$ or not, one can use the {\tt ``RosenfeldGroebner''} command in the Maple package {\tt DifferentialAlgebra}.
\end{Remark}

Below are examples of differential polynomials with vanishing order for each $m \in \N \cup \{\infty\}$.


\begin{Example} \label{EX:vanishing}
Consider the following AODE $$F = xyy''-yy'+x(y')^2 = 0.$$
A direct computation implies that for each $m \in \N$, we have $$\I_m(F) + \J_{2m}(F) \subseteq \langle c_0,c_1,\dots,c_{2 m + 1} \rangle.$$
Therefore, it follows from item 2 of Definition~\ref{DEF:globalvanishing} that $F$ has infinite vanishing order.
Note that $F=\frac{\partial F}{\partial y}=\frac{\partial F}{\partial y'}=\frac{\partial F}{\partial y''}=0$ has a common solution $y(x)=0$.
\end{Example}

\begin{Example}
Assume that $m \in \N$. 
Consider the AODE $$F=\frac{(y'+y)^2}{2}+x^{2m} = 0.$$
For $m=0$, it is straightforward to see that $F$ has vanishing order $0$. 

Let $m>0$.
By computation, we find that $\frac{\partial F}{\partial y'} = \frac{\partial F}{\partial y} = y'+y$. 
Therefore, we have that 
\begin{equation} \label{EQ:mnonvanishing}
\mathcal{I}_m(F)=\langle c_1+c_0,\dots,c_{m+1}+c_m \rangle.
\end{equation}
For each $k \geq 0$, it is straightforward to see that $\left((y'+y)^2\right)^{(k)}$ is a $\K$-linear combination of terms of the form $(y^{(i)} + y^{(i + 1)}) (y^{(j)} + y^{(j+1)})$ with $i+j=k$, and $i,j \geq 0$. 
Therefore, we conclude that for each $0 \leq k \leq m - 1$, the jet ideal $\mathcal{J}_{2 k}(F)$ is contained in $\mathcal{I}_m(F)$. 
It implies that $$\mathcal{I}_{k}(F) + \mathcal{J}_{2 k}(F) \subseteq \mathcal{I}_m(F) \ \ \text{ for } 0 \leq k \leq m - 1.$$
By \eqref{EQ:mnonvanishing} and the above formula, we have $$1 \notin \mathcal{J}_{2 k}(F) + \mathcal{I}_{k}(F)  \ \ \text{ for } 0 \leq k \leq m - 1.$$
Furthermore, we have that $$F^{(2m)}(0, c_1, \ldots, c_{1 + 2m}) \equiv (2m)! \mod \mathcal{I}_{m}(F).$$
Thus, it follows that $$1 \in \mathcal{I}_{m}(F) + \mathcal{J}_{2 m}(F)$$ and by definition, $F$ has vanishing order $m$.
\end{Example}


\begin{Theorem} \label{THM:global}
Let $F \in \K[x]\{y\}$ be of order $n$ with vanishing order $m \in \N$ and $\textbf{c} \in \V(\mathcal{J}_{2m}(F))$.
\begin{enumerate}
	\item There exists $i \in \{0, \ldots, m \}$ such that $F$ has vanishing order $i$ at $\pi_{n+i}(\textbf{c})$.
	\item Let $M = \max\{2m+i,\textbf{q}_{F,\textbf{c}}\}$. 
	Then $\textbf{c}$ can be extended to a formal power series solution of $F(y) = 0$ if and only if it can be extended to a zero point of $\mathcal{J}_{M}(F)$.
	\item Let $$\mathcal{V}_{\textbf{c}}(F) = \pi_{n+M-i}(\{ \tilde{\textbf{c}} \in \V(\mathcal{J}_{M}(F)) ~|~ \pi_{n+2m}(\tilde{\textbf{c}}) = \textbf{c}\}).$$
	Then $\mathcal{V}_{\textbf{c}}(F)$ is an affine variety of dimension at most $\textbf{r}_{F,\textbf{c}}$. 
	Moreover, each point of it can be uniquely extended to a formal power series solution of $F(y) = 0$.
\end{enumerate}
Hence, the set of formal power series solutions of $F(y)=0$ around the origin is in bijection with 
the set $$\bigcup\limits_{\textbf{c} \in \V(\J_{2m}(F))} \mathcal{V}_{\textbf{c}}(F).$$
\end{Theorem}
\begin{proof}
1. Since $\textbf{c} \in \V(\mathcal{J}_{2m}(F))$ and $F$ has vanishing order $m$, it follows that there exists a minimal $i \in \{0, \ldots, m \}$ such that 
$$\mathcal{S}_{F,i}(\textbf{c}) \neq 0.$$ 
and $F^{(k)}(\textbf{c}) = 0$ for $k=0,\ldots,2i$. 
Taking into account of item~2 of Lemma~\ref{LEM:2mvanish},  we see that only the first $n+i+1$ coefficients of $\textbf{c}$ are relevant and therefore $F$ has vanishing order $i$ at $\pi_{n+i}(\textbf{c})$.

2. and 3. The proofs are literally the same as those in Theorem~\ref{THM:local}.
\end{proof}

As a consequence of Theorem~\ref{THM:global}, for every AODE of finite order, say $m$, Algorithm~\ref{ALGO:localnonvanishing} can be applied to every given initial tuple $\textbf{c} \in \K^{n+m+1}$. 
So we can determine whether there exists a formal power series solution of $F(y)=0$ extending $\textbf{c}$ or not and in the affirmative case, all formal power series solutions can be described in finite terms.

Note that in case that $F$ has infinite vanishing order, there may exist an initial tuple $\textbf{c}$ with arbitrary size such that the set of formal power series solutions of $F(y)=0$ extending $\textbf{c}$ can not be described by an algebraic variety as item 3 of Theorem~\ref{THM:global}.
For instance, let us consider the differential polynomial $F=xyy''+yy'-x(y')^2$ from Example~\ref{EX:vanishing}.
For every $k \geq 1$, the set of all formal power series solutions of $F=0$ extending the zero initial tuple $\textbf{c}=\textbf{0} \in \mathbb{K}^k$ is the set $\{0\} \cup \{x^{r},\,|\,r \in \mathbb{N},\, r \geq k\}$.

In contrast to most other approaches, our method can be used to find singular solutions of AODEs as the following example illustrated. 
Hence, we avoid the algorithmic problem described in Remark~\ref{rem-regular}.
\begin{Example} \label{EX:singularnonvanishing}
Consider the AODE $$F=y'^2+y'-2y-x=0.$$ 
From Proposition~\ref{PROP:specialnonvanishing} we know that $F$ has vanishing order of at most $1$. 

Let $\textbf{c} = (-\frac{1}{8}, - \frac{1}{2}, 0, c_3)$, where $c_3$ is an arbitrary constant in $\K$. 
It is straightforward to verify that $\textbf{c}$ is a zero point of $\mathcal{J}_{2}(F)$. 
Furthermore, we have that $F$ has vanishing order $1$ at $\tilde{\textbf{c}} = \pi_{2}(\textbf{c}) = (-\frac{1}{8}, - \frac{1}{2}, 0)$. 
Therefore, we also know that $F$ has indeed vanishing order equal to $1$. 
We find that $S_{F,k,1}(\tilde{\textbf{c}})=-2$ and $M=3$. 
From item 2 of Theorem~\ref{THM:global}, we know that $\textbf{c}$ can be extended into a formal power series solution of $F(y) = 0$ if and only if it can be extended to a zero point of $\mathcal{J}_{3}(F)$. 
This is the case exactly for $c_3=0$. 
Hence, $\mathcal{V}_{\textbf{c}}(F) = \{ \textbf{c} \}$ and we can use Theorem~\ref{THM:global} to extend $\textbf{c}$ uniquely to the solution $$y_1(x) = -\frac{1}{8} - \frac{1}{2} x.$$
It is straightforward to verify that $y_1(x)$ is a singular solution of $F(y) = 0$.

Similarly, let $\tilde{\textbf{c}} = (-\frac{1}{8}, - \frac{1}{2}, 1, c_3)$, where $c_3$ is an arbitrary constant in $\K$.  
Using item 2 of Theorem~\ref{THM:global}, we deduce that $\tilde{\textbf{c}}$ can be extended into a formal power solution of $F(y) = 0$ if and only if $c_3 = 0$. 
In the affirmative case, $\mathcal{V}_{\tilde{\textbf{c}}}(F) = \{ \tilde{\textbf{c}} \}$ and we find that $$y_2(x) = -\frac{1}{8} - \frac{1}{2} x + \frac{1}{2} x^2$$ is the corresponding solution. 

Actually, one can verify that $y_1(x),y_2(x)$ are all the formal power series solutions of $F(y) = 0$ with $[x^0] S_F(y) = 0$. 
Therefore, the set of formal power series solutions of $F(y) = 0$ is equal to $$\{y_1(x),y_2(x) \} \cup \mathcal{S},$$ where $$\mathcal{S} = \{ y \in \K[[x]] ~|~ F(y) = 0 \, \text{ and } \, [x^0] S_F(y) \neq 0 \},$$ which can be determined by Proposition~\ref{PROP:ImplicitFunctionTheorem}.
\end{Example}

\section*{Acknowledgements}
The authors would like to thank Gleb Pogudin and Fran\c{c}ois Boulier for useful discussions.


\end{document}